\documentclass{elsarticle}

\usepackage[utf8]{inputenc}

\usepackage{amssymb}
\usepackage{amsmath}
\usepackage{amsthm}

\usepackage{pgf,tikz}
\usetikzlibrary{arrows}

\numberwithin{equation}{section}
\theoremstyle{plain}
\newtheorem{thm}{Theorem}[section]
\newtheorem{lemma}[thm]{Lemma}
\newtheorem{prop}[thm]{Proposition}
\newtheorem{coroll}[thm]{Corollary}
\newtheorem{claim}[thm]{Claim}

\theoremstyle{definition}
\newtheorem{definition}[thm]{Definition}

\newtheorem{remark}[thm]{Remark}

\usepackage{xspace} 

\newcommand{\cl}[1]{\mbox{\ensuremath{\mathbf{#1}}}\xspace}

\newcommand{\NP}{\cl{NP}}

\def\rank{{\rm rank}}
\def\minfas{{\rm minfas}}
\def\Div{{\rm Div}}
\def\Chip{{\rm Chip}}
\def\deg{{\rm deg}}
\def\dist{{\rm dist}}

\makeatletter
\def\ps@pprintTitle{%
  \let\@oddhead\@empty
  \let\@evenhead\@empty
  \def\@oddfoot{\reset@font\hfil\thepage\hfil}
  \let\@evenfoot\@oddfoot
}
\makeatother

\begin{document}

\title{Chip-firing games on Eulerian digraphs and {\NP}-hardness of computing the rank of a divisor on a graph}

\author[anal]{Viktor Kiss}
\ead{kivi@cs.elte.hu}

\author[cs]{Lilla T\'othm\'er\'esz}
\ead{tmlilla@cs.elte.hu}

\address[anal]{Department of Analysis, E\"otv\"os Lor\'and University, P\'azm\'any P\'eter s\'et\'any 1/C, Budapest H-1117, Hungary}

\address[cs]{Department of Computer Science, E\"otv\"os Lor\'and University, P\'azm\'any P\'eter s\'et\'any 1/C, Budapest H-1117, Hungary}

\begin{abstract}
Baker and Norine introduced a graph-theoretic analogue of the Riemann-Roch theory.
A central notion in this theory is the rank of a divisor.
In this paper we prove that computing the rank of a divisor on a  graph is \NP-hard, even for simple graphs.

The determination of the rank of a divisor can be translated to a question about a
chip-firing game on the same underlying graph.
We prove the \NP-hardness of this question by relating chip-firing on directed and undirected graphs.
\end{abstract}


\maketitle

\section{Introduction}

The Riemann-Roch theory for graphs was introduced by Baker and Norine in 2007 as the discrete analogue of the Riemann-Roch theory for Riemann surfaces \cite{BN-Riem-Roch}.
They defined the notions \emph{divisor, linear equivalence} and \emph{rank} also in this combinatorial setting, and showed that the analogue of basic theorems as for example the Riemann-Roch theorem, remains true.
Theorems like Baker's specialization lemma \cite{specializ_lemma} 
establish a connection between the rank of a divisor on a graph and on a curve, which enables a rich interaction of the discrete and continuous theories.

A central notion in the Riemann-Roch theory is the rank of a divisor. The question whether the rank can be computed in polynomial time has been posed in several papers \cite{hladky,Manjunath,pot_theory}, originally attributed to H.~Lenstra. 

Let us say a few words about previous work concerning the computation of the rank. Hladk\'y, Kr\'al' and Norine \cite{hladky} gave a finite algorithm for computing the rank of a divisor on a metric graph. Manjunath \cite{Manjunath} gave an algorithm for computing the rank of a divisor on a graph (possibly with multiple edges), that runs in polynomial time if the number of vertices of the graph is a constant. It can be decided in polynomial time, whether the rank of a divisor on a graph is at least $c$, where $c$ is a constant \cite{pot_theory}. Computing the rank of a divisor on a complete graph can be done in polynomial time \cite{cori}. For divisors of degree greater than $2g-2$ (where $g$ is the genus of the graph), the rank can be computed in polynomial time \cite{Manjunath}. On the other hand, there is a generalized model in which deciding whether the rank of a divisor is at least zero is already \NP-hard \cite{lattice}.

Our main goal in this paper is to show that
computing the rank of a divisor on a graph is \NP-hard, even for simple graphs. This result implies also the \NP-hardness of computing the rank of a divisor on a tropical curve by \cite[Theorem 1.6]{Luo}. We also show that deciding whether the rank of a divisor on a graph is at most $k$ is in \NP. 

Our method is the following: We translate the question of
computing the rank of a divisor to a question about the chip-firing
game of Bj\"orner, Lov\'asz and Shor
using the duality between these frameworks discovered by Baker
and Norine \cite{BN-Riem-Roch}. We get that the following question is computationally
equivalent to the determination of the rank: Given an initial
chip-distribution on an (undirected) graph $G$, what is the minimum number of
extra chips we need to put on this distribution to make the game non-terminating.

We first prove the \NP-hardness of computing the minimum number of chips that enables a non-terminating game on a simple Eulerian digraph by showing that it equals to the number of arcs in a minimum cardinality feedback arc set. This result is mentioned in a note added in proof of \cite{BL92}, where only the larger or equal part is proved. Recently, Perrot and Pham \cite{Perrot} solved an analogous question in the abelian sandpile model, which is a closely related variant of the chip-firing game. Our result follows by applying their method to the chip-firing game.

Then we show that the second question (concerning chip-firing games on directed graphs) can be reduced to the first one (concerning undirected graphs).
In order to do so, to any Eulerian digraph and initial chip-distribution, we assign an undirected graph with a chip-distribution
 such that in the short run, chip-firing on the undirected graph imitates chip-firing on the digraph.

\section{Preliminaries}

\subsection{Basic notations}

Throughout this paper, \emph{graph} means a connected undirected graph that can have multiple edges but no loops. A graph is \emph{simple} if it does not have multiple edges. A graph is usually denoted by $G$.
The vertex set and the edge set of a graph $G$ are denoted by $V(G)$ and $E(G)$, respectively.
The degree of a vertex $v$ is denoted by $d(v)$, the multiplicity of the edge $(u,v)$ by $d(u,v)$. The \emph{Laplacian matrix} of a graph $G$ means the following matrix $L$:
$$L(i,j) = \left\{\begin{array}{cl} -d(v_i) & \text{if } i=j \\
        d(v_i, v_j) & \text{if } i\neq j.
      \end{array} \right.$$

\emph{Digraph} means a (weakly) connected directed
graph that can have multiple edges but no loops.  
We usually denote a digraph by $D$. The vertex set and edge
set are denoted by $V(D)$ and $E(D)$, respectively.
For a vertex $v$ the indegree and the outdegree of $v$ are
denoted by $d^-(v)$ and $d^+(v)$, respectively.
A digraph $D$ is \emph{Eulerian} if $d^+(v) = d^-(v)$
for each vertex $v \in V(D)$.
The \emph{head} of the directed edge $(u, v)\in E(D)$
is $v$, and the \emph{tail} of the edge is $u$.
The multiplicity of the directed edge $(u,v)$ is
denoted by $\overrightarrow{d}(u, v)$.
A digraph is \emph{simple} if $\overrightarrow{d}(u, v) \le 1$ for each pair of different 
vertices $u, v \in V(D)$. 

The Laplacian matrix of a digraph $D$ means the
following matrix $L$:
$$L(i,j) = \left\{\begin{array}{cl} -d^+(v_i) & \text{if } i=j \\
        \overrightarrow{d}(v_j, v_i) & \text{if } i\neq j.
      \end{array} \right.$$

An important notion concerning digraphs is the feedback arc set. It also plays a crucial role in this paper.
\begin{definition}
A \emph{feedback arc set} of a digraph $D$ is a set of edges
$F \subseteq E(D)$ such that the digraph $D'=(V(D), E(D) \setminus F)$
is acyclic. We denote
$$
  \minfas(D) = \min\{|F| : F \subseteq E(D) \text{ is a feedback
  arc set}\}.
$$
\end{definition}

Let $G$ be a graph. An \emph{orientation} of $G$ is a directed
graph $D$ obtained from $G$ by directing each edge.
We identify the vertices of $G$ with the corresponding vertices
of $D$. We denote the
indegree and the outdegree of a vertex $v \in V(G)$ in the
orientation $D$ by $d^-_D(v)$ and $d^+_D(v)$,
respectively.

For a graph $G$ let us denote by $\mathbf{0}_G$
the vector with each
coordinate equal to $0$, and by $\mathbf{1}_G$
the vector with each
coordinate equal to $1$,
where the coordinates are
indexed by the vertices of $G$. For a vertex $v$ of $G$ we
denote the characteristic vector of $v$ by $\mathbf{1}_v$.
We use the same notations for digraphs.

\subsection{Riemann-Roch theory on graphs}

In this section we give some basic definitions of the Riemann-Roch theory on graphs. 
The basic objects are called \emph{divisors}.
For a graph $G$, $\Div(G)$ is the free abelian group on the set of vertices of $G$. An element $f\in \Div(G)$ is called \emph{divisor}. We either think of a divisor $f\in \Div(G)$ as a function $f:V(G)\to \mathbb{Z}$, or as a vector $f \in \mathbb{Z}^{|V(G)|}$, where the coordinates are indexed by the vertices of the graph.

The \emph{degree} of a divisor is the following: $$\deg(f)=\sum_{v\in V(G)} f(v).$$

The following equivalence relation on $\Div(G)$ is called \emph{linear equivalence}:
For $f, g \in \Div(G)$, $f\sim g$ if there exists a $z\in \mathbb{Z}^{|V(G)|}$ such that $g = f + Lz$.

A divisor $f\in\Div(G)$ is \emph{effective}, if $f(v)\geq 0$
for each $v\in V(G)$.

\begin{definition}[The rank of a divisor]
  For a divisor $f\in \Div(G)$, the \emph{rank} of $f$ is
  \begin{equation*}
  \begin{split}
    \rank(f) = \min\{\deg(g) - 1 : \text{$g \in \Div(G)$, $g$ is effective, } \\
    \text{$\nexists h \in \Div(G)$ such that $h \sim f-g$ and $h$ is effective}\}.
  \end{split}
  \end{equation*}
  When we wish to emphasize the underlying graph, we write $\rank_G(f)$ instead of $\rank(f)$.
\end{definition}

\subsection{Chip-firing}

  It was noted already by Baker and Norine \cite{BN-Riem-Roch}, that there is a duality between divisors on graphs and the objects of the chip-firing game, as defined by Bj\"orner, Lov\'asz and Shor \cite{BLS91}. Using this duality we can translate some questions in divisor theory to questions in chip-firing. We would like to use the latter language in the article, so let us include here a short introduction to chip-firing.

\begin{remark}
  Often the term ``chip-firing game'' is used also in the setting of Baker and Norine, but for clarity, we only use this term for the game of \cite{BLS91}.
\end{remark}

  The theory of chip-firing games was developed both for
  graphs (see \cite{BLS91}) and digraphs (see \cite{BL92}).
  Although the notions of the Riemann-Roch theory on graphs are in duality with
  notions concerning the chip-firing game on undirected graphs, later on,
  we also need chip-firing on digraphs.

  The basic idea is that on each vertex of a graph (or digraph),
  there is a certain amount of chips. If a vertex has at least as many chips as its degree
  (in the directed case: as its outdegree), then it can be fired.
  In the undirected case, this means, that the vertex passes
  a chip to its neighbors along each edge incident to it, and so the number of chips on itself decreases by its degree.
  In the directed case, the fired vertex passes a chip along each outgoing edge.
  In fact, if we think of an undirected graph as a special digraph where we replace
  each edge by a pair of oppositely directed edges, then the two definitions coincide for these graphs.

  Now we give the exact definitions.

  Let $H$ be an undirected graph or a digraph.
  A \emph{chip-distribution}, or \emph{distribution}, is a function
  $x:V(H) \to \mathbb{Z}^+\cup \{0\}$. We sometimes say: vertex
  $v$ has $x(v)$ chips. We use the notation $|x|$ for the number of
  chips in the distribution $x$, i.e., $|x| = \sum_{v \in V(H)} x(v)$.
  We denote the set of chip-distributions on $H$ by $\Chip(H)$.

  \emph{Firing} a vertex $v$ means taking the
  new chip-distribution $x + L\mathbf{1}_v$ instead of $x$.
  Note that the Laplacian matrix $L$ is different in the
  undirected and in the directed case, and that in both cases $|x+L\mathbf{1}_v|=|x|$
  so a firing preserves the number of chips.
  A vertex $v \in V(H)$ is \emph{active}
  (with respect to $x$) if after firing it, it still has
  a nonnegative number of chips, i.e., in the undirected case
  if $x(v)\geq d(v)$, while in the directed case if
  $x(v) \ge d^+(v)$.
  The firing of a vertex $v \in V(H)$ is \emph{legal},
  if $v$ was active before the firing.
  A \emph{legal game} is a sequence of distributions in which
  every distribution is obtained from the previous one by a legal
  firing.
  A game terminates if there is no active vertex with respect to the
  last distribution.

  The following theorem was proved by Bj\"orner, Lov\'asz and Shor
  for undirected graphs and by Bj\"orner and Lov\'asz for digraphs.
  (Originally the theorem for the undirected case was proved for
  simple graphs, but the proof works also for graphs with multiple edges.)

  \begin{thm}[{\cite[Theorem 1.1]{BLS91}}, {\cite[Theorem 1.1]{BL92}}] \label{thm::vegesseg_kommutativ}
    Let $H$ be a graph or a digraph and let $x \in \Chip(H)$
    be a distribution. Then starting from $x$, either every legal game can be continued indefinitely, or every legal game terminates after the same number of moves with the same final distribution. Moreover, the number of times a given node is fired is the same in every legal game.
  \end{thm}

  Let us call a chip-distribution $x\in \Chip(H)$ \emph{terminating}, if every legal chip-firing game played starting from $x$ terminates, and call it \emph{non-terminating}, if every legal chip-firing game played starting from $x$ can be continued indefinitely. According to Theorem \ref{thm::vegesseg_kommutativ}, a chip-distribution is either terminating or non-terminating.

  It can easily be seen by the pigeonhole principle that
  if for a graph $G$ a distribution $x \in \Chip(G)$ has $|x| > 2|E(G)| - |V(G)|$
  then $x$ is non-terminating (see \cite{BLS91}). And similarly,
  for a digraph $D$ if a distribution $x \in \Chip(D)$
  has $|x| > |E(D)| - |V(D)|$ then $x$ is non-terminating (see \cite{BL92}).

  From this it follows that the following quantity, which measures how far a
  given distribution is from being non-terminating, is well defined.
  For a distribution $x \in \Chip(H)$, let
  $$
    \dist(x) = \min\{|y| : y \in \Chip(H), x + y \text{ is non-terminating}\}.
  $$
  We say that $\dist(x)$ is the
  \emph{distance} of $x$ from non-terminating distributions.

  Note that $\dist(\mathbf{0}_H)$ is exactly the minimum number of chips in a non-terminating distribution on the graph/digraph $H$.

\subsection{Chip-firing and the Riemann-Roch theory}
  Now we describe the duality between divisors on graphs and chip-distri\-bu\-ti\-ons discovered by Baker and Norine \cite{BN-Riem-Roch}.

  Let $G$ be a graph and let $K^+ = K^+_G$ be the chip-distribution with
  $K^+(v)=d(v) - 1$ for each $v \in V(G)$.

  For a divisor $f\in \Div(G)$ with $f(v)\leq d(v)-1$ for each $v\in V(G)$, we have $(K^+ - f)(v)\geq 0$ for each $v\in V(G)$, therefore $K^+ - f\in \Chip(G)$.
  We call $K^+ - f$ the \emph{dual pair} of $f$. Note that each chip-distribution is a dual pair of some divisor.

  \begin{prop}[{\cite[Corollary 5.4]{BN-Riem-Roch}}]
    \label{prop::dual_pair}
    For a divisor $f\in \Div(G)$ with $f(v)\leq d(v)-1$ for each $v\in V(G)$, there exists an effective divisor equivalent to $f$ if and only if $K^+ - f$ is a terminating distribution.
  \end{prop}
  \begin{remark}
    We could have defined the chip-firing game for not necessarily nonnegative distributions as well with the same rules (only active vertices can fire). In this case Theorem \ref{thm::vegesseg_kommutativ} would still hold, and we could have a dual pair for any divisor, but this is not necessary for our purposes.
  \end{remark}

  The following is a straightforward consequence of
  Proposition \ref{prop::dual_pair}.
  \begin{prop}
    \label{prop::rank_of_dual}
    Let $f \in \Div(G)$ be a divisor with $f(v)\leq d(v)-1$ for each $v\in V(G)$,
    and let
    $x \in \Chip(G)$ be its dual pair. Then $\rank(f) = \dist(x) - 1$.
  \end{prop}

\subsection{{\NP}-hardness results}

  In Section \ref{sec::dist_digraph}, based on recent results of Perrot
  and Pham \cite{Perrot}, we prove the following.

  \begin{thm} \label{thm::rang_eulerben_nehez}
    Given a digraph $D$, computing $\dist(\mathbf{0}_D)$ is \NP-hard,
    even for simple Eulerian digraphs.
  \end{thm}

  Using this result we prove the main theorem of this article.

  \begin{thm} \label{thm::chip_rang_NP-teljes}
    For a distribution $x\in \Chip(G)$ on a graph $G$, computing $\dist(x)$ is \NP-hard.
  \end{thm}

  The proof can be found in Section \ref{sec::rank_chip_dist_NP_hard}.
  As a corollary of the theorem and Proposition \ref{prop::rank_of_dual}, we get the following.

  \begin{coroll}
    For a divisor $f\in \Div(G)$ on a graph $G$, computing $\rank(f)$ is \NP-hard, even for a divisor $f$ with $f(v)\leq d(v)-1$ for every $v \in V(G)$.
  \end{coroll}

  In \cite{hladky}, Hladk\'y, Kr\'al' and Norine prove the following statement:
  \begin{prop}[{\cite[Corollary 22.]{hladky}}] \label{prop::hladky_elosztos}
    Let $f$ be a divisor on a graph $G$.
    Let $G'$ be the simple graph obtained from $G$ by subdividing
    each edge of $G$ by an inner point and
    let $f'$ be the divisor on $G'$ that agrees with $f$ on
    the vertices of $G$ and has value $0$ on new points.
    Then $\rank_G(f)=\rank_{G'}(f')$.
  \end{prop}

  By dualizing this statement, we get the following: For a distribution $x\in\Chip(G)$,
  if we get $x'\in \Chip(G')$ from $x$ so that we put $d(v) - 1 - 0 = 1$
  chip on each new vertex, and on the vertices of $G$, $x'$ agrees with $x$,
  then
  $\dist_G(x) = \rank_G(K^+_G - x) + 1 = \rank_{G'}(K^+_{G'} - x') + 1 =
  \dist_{G'}(x')$.

  When proving Theorem \ref{thm::chip_rang_NP-teljes}, we show a 
  somewhat stronger statement: By Remark
  \ref{rem::dist_NP_erosites}, computing $\dist$ is \NP-hard even 
  for graphs with $|E(G)|\leq 9|V(G)|^5$. For such a $G$, 
  $|V(G')|\leq |V(G)|+|E(G)|\leq 10|V(G)|^5$. Hence $G'$ and $x'$ can be computed in polynomial time for such a graph $G$ and $x\in \Chip(G)$, giving the following corollary.

  \begin{coroll}
    For a distribution $x\in \Chip(G)$ on a simple graph $G$, computing $\dist(x)$ is \NP-hard.
  \end{coroll}
  Using Proposition \ref{prop::rank_of_dual} again, 
  we have the following. 
  \begin{coroll}
    For a divisor $f\in \Div(G)$ on a simple graph $G$, computing $\rank(f)$ is \NP-hard.
  \end{coroll}

  Using a result of \cite{Luo}, we get that computing the rank of divisors is also \NP-hard for so called tropical curves.
  Informally, a metric graph is a graph, where each edge has a positive length, and we consider our graph to be a metric space (the inner points of the edges are also points of this metric space). Tropical curve is more general in that we also allow some edges incident with vertices of degree one to have infinite length.
  A divisor on a tropical curve is an integer-valued function on the curve with only finitely many nonzero values. The notions of the degree of a divisor, linear equivalence, effective divisor and the rank can be defined as well, see \cite{hladky}.

  A metric graph $\Gamma$ corresponds to the graph $G$, if $\Gamma$ is obtained from $G$ by assigning some positive length to each edge.
  \begin{thm}[{\cite[Theorem 1.6]{Luo}}] \label{thm::metrikus_rang_rang}
    Let $f$ be a divisor on a graph $G$, and $\Gamma$ be a metric graph corresponding to $G$. Then $\rank_G(f)=\rank_{\Gamma}(f)$.
  \end{thm}

  As a metric graph is a special tropical curve, we get the following corollary:

  \begin{coroll}
    For a tropical curve $\Gamma$, $f\in \Div(\Gamma)$,
    computing $\rank(f)$ is \NP-hard.
  \end{coroll}

From the positive side, we show the following:

\begin{prop}
 Deciding whether for a given divisor $f$ on a graph $G$, and integer $k$, $\rank(f)\leq k$ is in \NP.
\end{prop}

\begin{proof}
 For an input $(f,k)$ with $\rank(f)\leq k$, our witness is the divisor $g\geq 0$ such that $\rank(f-g)=-1$, and $\deg(g)\leq k+1$ (such a $g$ exists because $\rank(f)\leq k$).

 First, we need to check that $g$ can be given so that it has size polynomial in the size of $(f,k)$. As $\deg(g)\leq k+1$, and $g\geq 0$, we have $g(v)\leq k+1$ for each vertex $v$. Therefore, the size of $g$ is at most $O(|V(G)|\cdot \log k)$.

 On the other hand, it can be checked in polynomial time if $\rank(f-g)=-1$ \cite{pot_theory}, and also whether $\deg(g)\leq k+1$.
\end{proof}

By applying Proposition \ref{prop::rank_of_dual}, deciding whether for a given chip-distribution $x$, and integer $k$, $\dist(x)\leq k$ is also in \NP.

\section{Minimal non-terminating distributions on Eulerian digraphs}
\label{sec::dist_digraph}
  In this section we prove
  Theorem \ref{thm::rang_eulerben_nehez}, i.e., 
  that computing the minimum number of chips in a non-terminating distribution is \NP-hard for a simple Eulerian digraph $D$.

  We use the method of Perrot and Pham. In the paper \cite{Perrot},
  they prove the \NP-hardness of an analogous question in the abelian sandpile model,
  which is a closely related variant of the chip-firing game.

  Using the ideas of \cite{Perrot}, we first give a formula for the minimum
  number of chips in a non-terminating distribution on an Eulerian digraph.
  As a motivation, let us have a look at the analogous question on undirected graphs, which was solved by Bj\"orner, Lov\'asz and Shor.
  
    \begin{thm}[{\cite[Theorem 2.3]{BLS91}}]
      \label{thm::also_korlat}
      Let $G$ be a graph. Then $\dist(\mathbf{0}_G)=|E(G)|$.
    \end{thm}
    We sketch the proof as a motivation for the directed
    case.
    \begin{proof}
      First we prove the following useful lemma.
      \begin{lemma}[\cite{BLS91}]
        \label{lem::chip_dist_over_acyclic}
        Let $D$ be an acyclic orientation of $G$ and let
        $x\in \Chip(G)$ be a distribution with $x(v) \ge d^-_D(v)$
        for each $v \in V(G)$. Then $x$ is non-terminating.
      \end{lemma}
      \begin{proof}
        Since the orientation is acyclic, there is a sink, i.e., a
        vertex $v_0 \in V(G)$ with $d(v_0) = d^-_D(v_0) \leq x(v_0)$.
        Hence $v_0$ is active with respect to $x$. Fire $v_0$ and
        denote the resulting distribution by $x'$.
        Reverse the direction of the edges incident to $v_0$ and
        denote the resulting directed graph by $D'$. It is easy to
        see that $D'$ is acyclic and $d^-_{D'}(v) \le x'(v)$ for each
        $v \in V(G)$. Hence we can repeat the above argument. This
        shows that the distribution $x$ is indeed non-terminating.
      \end{proof}
      Now taking an acyclic orientation $D$ of $G$ and setting
      $x(v) = d^-_D(v)$ for each $v \in V(G)$ we have a distribution
      with $|x| = |E(G)|$ that is non-terminating from the lemma.
      This shows that $\dist(\mathbf{0}_G)\leq |E(G)|$.

      For proving $\dist(\mathbf{0}_G)\geq |E(G)|$, take a non-terminating distribution
      $x \in \Chip(G)$. It is enough to show that $|x|\geq |E(G)|$.
      Since in a non-terminating game every vertex is fired infinitely often
      (see \cite[Lemma 2.1]{BLS91}), after finitely many firings,
      every vertex of $G$ has been fired at least once.
      Let $x'$ be the distribution at such a moment.
      Then $|x| = |x'|$.
      Let $D$ be the orientation of $G$ that we get by
      directing each edge toward the vertex whose last firing occurred
      earlier. It is straightforward to check that
      $x'(v) \ge d^-_D(v)$ for each $v \in V(G)$.
      This fact implies that $|x| = |x'| \ge |E(G)|$,
      completing the proof.
    \end{proof}

    Now let us consider Eulerian digraphs. 

    \begin{thm}
      \label{thm::dist = minfas}
      Let $D$ be an Eulerian digraph.
      Then $\dist(\mathbf{0}_D) = \minfas(D)$.
    \end{thm}

    This theorem is already stated in a note added in 
    proof of \cite{BL92}, but there only the direction 
    $\dist(\mathbf{0}_D) \ge \minfas(D)$ is proved.
    We give a proof following ideas of Perrot and Pham \cite{Perrot}.
    The idea of the proof can be 
    thought of as the generalization of the idea of the proof of 
    Theorem \ref{thm::also_korlat}.
    For proving Theorem \ref{thm::dist = minfas} we need a classical result about chip-firing on an Eulerian digraph:

    \begin{prop}[{\cite[Lemma 2.1]{BL92}}] \label{prop::euler_vegt_mind_vegt_sokszor_lo}
      On an Eulerian digraph $D$ if a chip-distribution is non-terminating then in any legal game every vertex is fired infinitely often.
    \end{prop}
    \begin{proof}[Proof of Theorem \ref{thm::dist = minfas}]
      The key lemma is the following observation of Perrot and Pham; they proved it for the recurrent configurations of the abelian sandpile model, but the two models are very closely related.

      \begin{lemma}[{\cite{Perrot}}] \label{lemma::euler_ir_felett_vegtelen}
        Let $F\subseteq E(D)$ be a minimum cardinality feedback arc set. Denote by $d^+_F(v)$ and $d^-_F(v)$ the outdegree and indegree of a vertex $v$ in the digraph $D_F=(V(D), F)$.
        Then a distribution $x\in \Chip(D)$ satisfying
        \begin{equation}
          \label{eq::feedback_felett}
          x(v)\geq d^-_F(v) \text{ for every $v\in V(D)$}
        \end{equation}
        is non-terminating.
      \end{lemma}
      \begin{proof}
        First we prove that if $F$ is a minimum cardinality
        feedback arc set then there exists a vertex $v \in V(D)$ such
        that among the edges incident to $v$, $F$ contains exactly
        the in-edges of $v$.

        Let $A = E(D) \setminus F$. From the definition of feedback
        arc set, $D_A = (V(D), A)$ is an acyclic graph.
        Therefore, it has a source $v_0$.
        We claim that no out-edge of $v_0$ is in $F$.
        Indeed, if some out-edges of $v_0$ would be in $F$, removing
        them from $F$ would mean adding some out-edges to the source $v_0$ in
        $D_A$, which cannot create a cycle. So we could get a smaller
        feedback arc set.

        The fact that $v_0$ is a source in $D_A$ means that all the
        in-edges of $v_0$ are in $F$. Hence from the edges
        incident to $v_0$, $F$ contains exactly the in-edges of
        $v_0$.

        Now take such a vertex $v_0$.
        From \eqref{eq::feedback_felett}, the choice of $v_0$ and
        the fact that $D$ is Eulerian, we have that
        $x(v_0)\geq d^-_F(v_0)=d^-(v_0)=d^+(v_0)$,
        therefore $v_0$ is active with respect to $x$.
        Fire $v_0$. Let $x'$ be the resulting distribution.
        We show, that we can modify the feedback arc set $F$, such
        that for the new feedback arc set $F'$ we have
        $x'(v)\geq d^-_{F'}(v)$ for every $v\in V(D)$.

        Let $F'$ be the set of arcs obtained from $F$ by removing the
        in-edges of $v_0$ and adding the out-edges of $v_0$
        (see Figure \ref{fig::feedback}).
        Then $D_{A'} = (V(D), A')$ with $A'=E(D)\setminus F'$ is
        acyclic, since the new edges are all incident to a sink in
        $D_{A'}$.
        Moreover, since $D$ is Eulerian, and from the choice of $v_0$, we have $|F'|=|F|$, hence
        $F'$ is also a feedback arc set of minimum cardinality.
        It is straightforward to check that indeed
        $x'(v)\geq d^-_{F'}(v)$ for every $v \in V(D)$.

        So we are again in the starting situation, which shows that
        $x$ is indeed non-terminating.
       \end{proof}
\begin{figure}[ht]
\begin{center}
\begin{tikzpicture}[->,>=stealth',auto,scale=2.3,
                    thick,every node/.style={circle,draw,font=\sffamily\small}]

  \node[label=left:0] (1) at (0, 1) {};
  \node[label=left:0] (2) at (0, -1) {};
  \node[label=below:0] (3) at (-1, 0) {};
  \node[label=below:2] (4) at (1, 0) {};
  \path[every node/.style={font=\sffamily\small}]
    (1) edge [dashed] node {} (4)
    (3) edge node {} (1)
    (2) edge node {} (3)
    (4) edge node {} (2)
    (4) edge [bend left=20] node {} (3)
    (3) edge [dashed,bend left=20] node {} (4);
\end{tikzpicture}
\hspace{0.2cm}
\begin{tikzpicture}[->,>=stealth',auto,scale=2.3,
                    thick,every node/.style={circle,draw,font=\sffamily\small}]

  \node[label=left:0] (1) at (0, 1) {};
  \node[label=left:1] (2) at (0, -1) {};
  \node[label=below:1] (3) at (-1, 0) {};
  \node[label=below:0] (4) at (1, 0) {};
  \path[every node/.style={font=\sffamily\small}]
    (1) edge node {} (4)
    (3) edge node {} (1)
    (2) edge node {} (3)
    (4) edge [dashed] node {} (2)
    (4) edge [dashed,bend left=20] node {} (3)
    (3) edge [bend left=20] node {} (4);
\end{tikzpicture}
\caption{An example for simultaneously firing a vertex and changing the feedback arc set. The arcs of the feedback arc sets are drawn by dashed lines.}
\label{fig::feedback}
\end{center}
\end{figure}
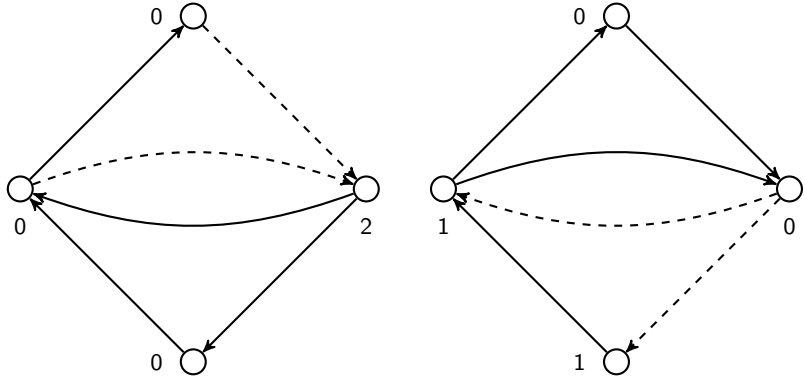

         Now take a feedback arc set $F$ of minimum cardinality,
         and let $x(v)=d^-_F(v)$ for every $v\in V(D)$.
         Then $|x|=|F|=\minfas(D)$, and from the lemma, $x$ is non-terminating.
         This proves that $\dist(\mathbf{0}_D)\leq \minfas(D)$.

         The direction $\dist(\mathbf{0}_D)\geq \minfas(D)$ is shown in the note 
         added in proof of \cite{BL92} for a general digraph, 
         however, as we need implications of its idea, we also include this part of the proof.         
         
         Take a non-terminating
         distribution $x$. It is enough to prove that $|x|\geq \minfas(D)$.

         Let us play a chip-firing game with initial distribution $x$.
         Proposition \ref{prop::euler_vegt_mind_vegt_sokszor_lo} says that after finitely many steps, every vertex
         has fired. Play until such a moment, and let the distribution at that moment be $x'$.

         Let $A$ be the following set of edges:
         $$
         A=\{(u,v) \in E(D): \textrm{ the last firing of $u$ preceeds the last firing of $v$}\}.
         $$
         As every vertex has fired, $A$ is well defined.
         Let $v_1,v_2, \dots v_{|V(D)|}$ be the ordering of the vertices by the time of their last firing.
         Then $v_1,v_2, \dots v_{|V(D)|}$ is a topological order of $D_A=(V(D),A)$, so $D_A$ is acyclic, hence $F=E(D)\setminus A$ is a feedback arc set.
         We show that $x'(v)\geq d^-_F(v)$ for every $v\in V(D)$.
         For $1 \le i \le |V(D)|$ the vertex $v_i$ has $d^-_F(v_i)=\sum_{j>i}\overrightarrow{d}(v_j,v_i)$.
         After its last firing, $v_i$ had a nonnegative number of chips.
         Since then, it kept all chips it received. And as $v_{i+1}, \dots, v_{|V(D)|}$
         all fired since the last firing of $v_i$, it received at least $\sum_{j>i}\overrightarrow{d}(v_j,v_i)=d^-_F(v_i)$ chips. So indeed, we have $x'(v_i)\geq d^-_F(v_i)$.

         Therefore $|x|=|x'| \ge |F|\geq \minfas(D)$.
      \end{proof}

      Note that in the above setting, starting from $x'$, then firing the vertices in the order $v_1,v_2,\dots, v_{|V(D)|}$ (once each) is a legal game. Indeed, we proved that $x'(v_i)\geq d^-_F(v_i)=\sum_{j>i}\overrightarrow{d}(v_j,v_i)$. After firing $v_1, \dots v_{i-1}$, the vertex $v_i$ receives $\sum_{j<i}\overrightarrow{d}(v_j,v_i)$ more chips, so it indeed becomes active
      ($d^+(v_i)=d^-(v_i)=\sum_{j\neq i}\overrightarrow{d}(v_j,v_i))$ as we did not allow loops).

      We need this observation in the next section, so we state it as a proposition:

      \begin{prop} \label{prop::euler_ha mar mindenki lott}
        In a chip-firing game on an Eulerian digraph $D$, if at some moment every vertex has already fired then there is an order of the vertices in which they can be legally fired once each, starting from that moment.\qed
      \end{prop}

      It is worth noting that on an Eulerian digraph if starting from an initial
      distribution $x$ we fired each vertex exactly once, then we get back to distribution $x$:
      each vertex $v$ gave and received $d^-(v)=d^+(v)$ chips.

      Finally, we prove Theorem \ref{thm::rang_eulerben_nehez}.

      \begin{proof}[Proof of Theorem \ref{thm::rang_eulerben_nehez}]
        Perrot and Pham proved that computing $\minfas(D)$ for a 
        simple Eulerian digraph $D$ is \NP-hard \cite[Theorem 2]{Perrot}, 
        by reducing it to the \NP-hardness of computing $\minfas(D)$ 
        for general digraphs.
        From this, and from Theorem \ref{thm::dist = minfas}, the 
        statement follows.
      \end{proof}

\section{The distance from non-terminating distributions is {\NP}-hard on graphs}
  \label{sec::rank_chip_dist_NP_hard}
  In this section we prove Theorem \ref{thm::chip_rang_NP-teljes}, the main theorem of this article.
  In our proof of the \NP-hardness, we rely on the fact
  that a terminating chip-firing game on an Eulerian digraph $D$
  terminates after at most $2|V(D)|^2 |E(D)| \Delta(D)$ steps
  (see \cite[Corollary 4.9]{BL92}), where $\Delta(D)$ denotes 
  the maximum of all the indegrees and the outdegrees of $D$, i.e.,
  $\Delta(D) = \max_{v \in V(D)}\max\{d^-(v), d^+(v)\}$.
  With this in mind, we define the following transformation:

  \begin{definition}
    Let $\varphi$ be the following transformation, assigning an undirected graph $G=\varphi(D)$ to any digraph $D$:

    Split each directed edge by an inner point, and substitute the tail segment by $M=8|V(D)|^2 |E(D)| \Delta(D)$ parallel edges. Then forget the orientations.

    We maintain the effect of the transformation by a bijective function $\psi: (V(D)\cup E(D)) \to V(\varphi(D))$:

    For a vertex $v\in V(D)$ let $\psi(v)$ be the corresponding vertex of $\varphi(D)$. For an edge $e \in E(D)$, let $\psi(e)$ be the vertex with which we have split $e$.
  \end{definition}

  Then the degrees in $\varphi(D)$ are the following:
  \begin{equation*}
      d(v) = \left\{\begin{array}{cl} d^+\left(\psi^{-1}(v)\right)\cdot M + d^-(\psi^{-1}(v)) & \text{if } \psi^{-1}(v) \in V(D) \\
        M + 1 & \text{if } \psi^{-1}(v) \in E(D).
      \end{array} \right.
  \end{equation*}

\begin{figure}[ht]
\begin{center}
\begin{tikzpicture}[->,>=stealth',auto,scale=2.5,
                    thick,every node/.style={circle,draw,font=\sffamily\small}]

  \node (1) at (0, 1) {\small $v_1$};
  \node (2) at (0, -1) {\small $v_2$};
  \node (3) at (-1, 0) {\small $v_3$};
  \node (4) at (1, 0) {\small $v_4$};
  \path[every node/.style={font=\sffamily\small}]
    (1) edge node {$e_1$} (4)
    (3) edge node {$e_2$} (1)
    (2) edge node {$e_3$} (3)
    (4) edge node {$e_4$} (2)
    (4) edge [bend left=20] node {$e_5$} (3)
    (3) edge [bend left=20] node {$e_6$} (4);
\end{tikzpicture}
\begin{tikzpicture}[auto,scale=2.5,
                    thick,every node/.style={circle,draw,font=\sffamily\small}]

  \node (1) at (0, 1) {\tiny $\psi(v_1)$};
  \node (2) at (0, -1) {\tiny $\psi(v_2)$};
  \node (3) at (-1, 0) {\tiny $\psi(v_3)$};
  \node (4) at (1, 0) {\tiny $\psi(v_4)$};
  \node (41) at (0.5, 0.5) {\tiny $\psi(e_1)$};
  \node (13) at (-0.5, 0.5) {\tiny $\psi(e_2)$};
  \node (32) at (-0.5, -0.5) {\tiny $\psi(e_3)$};
  \node (24) at (0.5, -0.5) {\tiny $\psi(e_4)$};
  \node (34) at (0, -0.25) {\tiny $\psi(e_5)$};
  \node (43) at (0, 0.25) {\tiny $\psi(e_6)$};

  \path[every node/.style={font=\sffamily\small}]
    (4) edge (41)
    (41) edge [bend right=10] (1)
    (41) edge [bend left=10] (1)
    (41) edge (1)
    (1) edge (13)
    (13) edge [bend right=10] (3)
    (13) edge [bend left=10] (3)
    (13) edge (3)
    (3) edge (32)
    (32) edge [bend right=10] (2)
    (32) edge [bend left=10] (2)
    (32) edge (2)
    (2) edge (24)
    (24) edge [bend right=10] (4)
    (24) edge [bend left=10] (4)
    (24) edge (4)
    (3) edge (34)
    (34) edge [bend right=10] (4)
    (34) edge [bend left=10] (4)
    (34) edge (4)
    (4) edge (43)
    (43) edge [bend right=10] (3)
    (43) edge [bend left=10] (3)
    (43) edge (3);
\end{tikzpicture}
\caption{A schematic picture for a digraph $D$ and the corresponding $\varphi(D)$. In the reality the multiple edges should be 1536-fold.}
\end{center}
\end{figure}
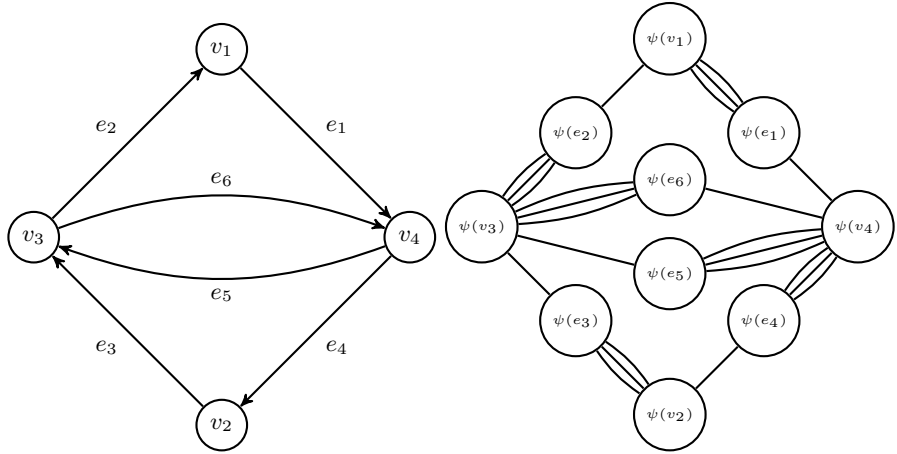

  Let us define a certain chip-distribution on the graph $\varphi(D)$:

  \begin{definition}[base-distribution]
    Let $base_D \in \Chip(\varphi(D))$ on a vertex $v \in V(\varphi(D))$ be the following:
    \begin{equation*}
    \begin{split}
      base_D(v) = \left\{\begin{array}{cl} d^+\left(\psi^{-1}(v)\right)\cdot M & \text{if } \psi^{-1}(v) \in V(D) \\
        M/2 & \text{if } \psi^{-1}(v) \in E(D).
      \end{array} \right.
    \end{split}
    \end{equation*}
  \end{definition}

  The key lemma in our proof of Theorem \ref{thm::chip_rang_NP-teljes} is the following:
  \begin{lemma}
    For an Eulerian digraph $D$, $\dist_D(\mathbf{0}_D) = \dist_{\varphi(D)}(base_D)$.
  \end{lemma}
  \begin{proof}
    Let $G = \varphi(D)$.
    First we show that $\dist_D(\mathbf{0}_D) \ge \dist_{G}(base_D)$.

    Let $x \in \Chip(D)$ be a non-terminating chip-distribution such that $|x|$ is minimal.
    We can assume that there is an order
    of the vertices of $D$ such that from initial distribution
    $x$ we can fire the vertices in that order (once each).
    Otherwise, from Proposition \ref{prop::euler_vegt_mind_vegt_sokszor_lo} we can play a
    chip-firing game from $x$ until each vertex has fired. Denoting
    the distribution at that moment by $x'$, from Proposition
    \ref{prop::euler_ha mar mindenki lott} for $x'$ there is such an order. As firing
    does not change the number of chips in the game, $|x'|$ is still minimal, so we can
    substitute $x$ with $x'$.

    Let $y\in\Chip(G)$ be the distribution ``$x + base_D$'', i.e., for a vertex
    $v \in V(D)$ let $y(\psi(v)) = x(v) + base_D(\psi(v))$ and for an
    edge $e \in E(D)$ let $y(\psi(e)) = base_D(\psi(e))$.
    Since $y(w) \ge base_D(w)$ for each $w \in V(G)$ and
    $|y - base_D| = |x| = \dist_D(\mathbf{0}_D)$,
    it is enough to show that $y$ is non-terminating.

    For that, it is enough to show that we can fire each vertex of $G$
    exactly once in some order. Then each vertex $w \in V(G)$ gives and receives $d(w)$ chips,
    so we get back to the distribution $y$ and can repeat this period indefinitely.

    To get such an order of the vertices of $G$, we will play the chip-firing game simultaneously on $D$ and $G$.

    To firing a vertex $v$ in $D$, let the corresponding firings in
    $G$ be: Fire $\psi(v)$, then fire $\psi(e)$ for every out-edge
    $e$ of $v$ (in some order).

    \begin{claim}
      If a sequence of firings of length $k \leq M/2$ on
      $D$ with initial distribution $x$ is legal then the sequence of the
      corresponding firings on $G$ with
      initial distribution $y$ is also legal. Moreover, if we
      denote the resulting distribution on $D$ by $\tilde{x}$ and on $G$ by
      $\tilde{y}$ then
      \begin{equation}
        \label{eq::vertex}
        \tilde{y}(\psi(v)) = \tilde{x}(v) + d^+(v) \cdot M
          \text{ for each $v \in V(D)$}
      \end{equation}
      and
      \begin{equation}
        \label{eq::edge}
        M/2 - k \le \tilde{y}(\psi(e)) \le M/2 + k \text{ for each $e \in E(D)$.}
      \end{equation}
    \end{claim}
    \begin{proof}
      We show this by induction on $k$. For $k = 0$ this is trivial.
      Take a sequence of firings of length $k \leq M/2$ and assume that the
      claim holds for $k - 1$. Denote the distribution on $D$ after
      the first $k - 1$ firings by $x'$ and the corresponding
      distribution on $G$ by $y'$. Assume that the vertex $v$ is the
      last to be fired on $D$. Hence $v$ is active with respect
      to $x'$. Denote the distribution after firing $v$ by
      $x''$. Vertex $\psi(v)$ is active with respect to $y'$,
      since using \eqref{eq::vertex} of the induction hypothesis, the fact that $v$
      is active with respect to $x'$ and that $D$ is Eulerian, we get that
      $y'(\psi(v)) = x'(v) + d^+(v) \cdot M \ge
      d^+(v) + d^+(v) \cdot M =
      d^-(v) + d^+(v) \cdot M = d(\psi(v))$.
      Fire $\psi(v)$. Now for each out-edge $e$ of $v$ the vertex
      $\psi(e)$ is active, since using \eqref{eq::edge} of the induction hypothesis,
      it has at least
      $M + y'(\psi(e)) \ge M + M/2 - (k - 1) \geq M + 1 = d(\psi(e))$ chips.
      Fire these vertices in an arbitrary order.
      (Firing one leaves the others active.)
      Denote by $y''$ the resulting distribution.
      It is easy to check that the distributions $x''$ and $y''$ satisfy
      conditions \eqref{eq::vertex} and \eqref{eq::edge}.
    \end{proof}

    We have chosen the distribution $x$ such that we can fire the vertices
    of $D$ in some order (once each) with initial distribution $x$.
    This is a legal sequence of firings of length $|V(D)|<M/2$.
    According to the previous
    claim, the sequence of the corresponding firings on $G$ is also legal.
    Moreover, on $G$ we also fire each
    vertex exactly once. This finishes the proof of the direction
    $\dist_D(\mathbf{0}_D) \ge \dist_G(base_D)$.

    Now we prove that $\dist_D(\mathbf{0}_D) \le \dist_G(base_D)$.
    For this, let $y \in \Chip(G)$ be a minimal non-terminating
    chip-distribution with $base_D(w) \le y(w)$ for each $w \in V(G)$.
    Let $x(v) = y(\psi(v)) - base_D(\psi(v))$ on each $v\in V(D)$.
    It is enough to show that $x$ is non-terminating.

    First note that $\dist_D(\mathbf{0}_D) \le |E(D)| - |V(D)| + 1$,
    since having a chip-distribu\-ti\-on
    with at least $|E(D)|-|V(D)|+1$ chips, at every stage of the game
    at least one of the vertices has the sufficient
    number of chips to fire.
    Consequently, using also the first part of the lemma, we have
    that $|y - base_D| = \dist_G(base_D) \le \dist_D(\mathbf{0}_D) \le
    |E(D)|-|V(D)| + 1 \le \frac{1}{8}M$.

    Now we play the game on $G$ and $D$ simultaneously
    from initial distributions $y$ and $x$, respectively,
    in the following way. Let a step be the following:
    Choose a vertex $v \in V(D)$ for which $\psi(v)$ can fire.
    On $G$ fire $\psi(v)$, then for every out-edge $e$ of $v$,
    fire $\psi(e)$. On $D$ fire $v$.

    We show that
    for $\frac{3}{8}M \ge
    2|V(D)|^2 |E(D)| \Delta(D) + 1$ steps we
    can play this legally on both graphs.
    Note first that for an edge $e$ of $D$, the change of the
    number of chips on $\psi(e)$ is at most one after each step.
    Hence at the beginning of a step a vertex of $G$ of
    the form $\psi(e)$ can have at most
    $M/2 + |y - base_D| + \frac{3}{8}M \leq
    M/2 + \frac{1}{8}M + \frac{3}{8}M < M + 1 = d(\psi(e))$ chips, so
    it cannot be fired. It also follows from this that on every such
    vertex the number of chips is positive, since it is at least
    $M/2 - \frac{3}{8}M > 0$.
    But $y$ is a non-terminating distribution, hence at the beginning of a
    step we can find an active vertex, which therefore must be of the form $\psi(v)$ with $v \in V(D)$.
    After firing $\psi(v)$, $\psi(e)$ becomes active
    for every out-edge $e$ of $v$, since $\psi(e)$
    had a positive number of chips at the beginning of the step,
    and received $M$ chips. Hence on $G$ we can play in the desired way for $\frac{3}{8}M$ steps.

    For the initial distributions, we have
    $y(\psi(v)) = d^+(v) \cdot M + x(v)$ for each $v\in V(D)$,
    so a vertex $v \in V(D)$ is active with respect to $x$
    if and only if $\psi(v)$ is active with respect to $y$.
    Let $x'$ be the distribution on $D$ and $y'$ the distribution
    on $G$ at the end of an arbitrary (but at most $\frac{3}{8}M^\text{th}$) step.
    Then it can be shown by induction
    that $y'(\psi(v)) = d^+(v) \cdot M + x'(v)$
    for each $v \in V(D)$. So in each step we have that
    a vertex $v \in V(D)$ is active if and only if $\psi(v)$ is active.

    Hence for $\frac{3}{8}M$ steps, the corresponding game on $D$ is also legal.
    This means that there is a chip-firing game of length at least
    $\frac{3}{8}M \ge 2|V(D)|^2 |E(D)| \Delta(D) + 1$ on $D$ with initial distribution $x$, which by
    \cite[Corollary 4.9]{BL92} implies that the distribution $x$ is non-terminating.
    This finishes the proof.
  \end{proof}

  For a general digraph, the construction of the proof imitates the following game: If a vertex $v$ fires, each of its out-neighbors $u$ receives $\overrightarrow{d}(vu)$ chips, but the number of chips on $v$ decreases by the in-degree of $v$. This modification of the chip-firing game has been studied by Asadi and Backman \cite{asadi}.

  \begin{proof}[Proof of Theorem \ref{thm::chip_rang_NP-teljes}]
    The theorem follows from Theorem \ref{thm::rang_eulerben_nehez} 
    and the previous lemma.
  \end{proof}

\begin{remark} \label{rem::dist_NP_erosites}
  For a simple Eulerian digraph $D$, one has $$|E(\varphi(D))|\leq |E(D)|\cdot 9|V(D)|^3|E(D)|\leq 9|V(\varphi(D))|^5,$$ therefore the computation of $\dist$ is \NP-hard even for graphs with $|E(G)|\leq 9|V(G)|^5$.
\end{remark}

\section{Polynomial time computability in a special case}
  In this section we consider undirected graphs, and observe that for
  chip-distributions that are in a sense ``small'', computing the
  distance from non-terminating distributions
  can be done in polynomial time. Moreover, for these distributions,
  the distance from non-terminating distributions only depends on the number of edges of the graph and the
  number of chips in the distribution.
  
  The corollaries of this observation for the case of divisors 
  give a special case of the Riemann-Roch theorem.

  Recall that Theorem \ref{thm::also_korlat} stated that
  $\dist(\mathbf{0}_G)=|E(G)|$ for any
  undirected graph $G$. We would like to generalize this statement for
  ``small enough'' distributions.
  We say that a distribution $x \in \Chip(G)$ is \emph{under an
  acyclic orientation}, if there exists an acyclic orientation $D$ of $G$ such that
  $x(v) \le d^-_D(v)$ for each $v \in V(G)$.

  \begin{prop} \label{prop::acikl_alatt}
    Let $G$ be a graph and let $x \in \Chip(G)$ be a
    distribution. If $x$ is
    under an acyclic orientation then $\dist(x) = |E(G)| - |x|$.
  \end{prop}
  \begin{proof}
    From Theorem \ref{thm::also_korlat}, a
    non-terminating distribution has at least $|E(G)|$ chips, therefore
    $\dist(x) \ge |E(G)| - |x|$.

    For the other direction, let $D$ be an acyclic orientation of $G$
    with $x(v) \le d^-_D(v)$ for each $v \in V(G)$. Let $y$ be the
    distribution on $G$ corresponding to the indegrees of the
    orientation, i.e., $y(v) = d^-_D(v)$ for each $v \in V(G)$.
    Then, using Lemma \ref{lem::chip_dist_over_acyclic}, $y$ is
    non-terminating, moreover $|y| = |E(G)|$ and $y(v) \ge x(v)$ for
    each $v \in V(G)$.
    Hence $\dist(x) \le |y - x| = |E(G)| - |x|$.
    This completes the proof of the proposition.
  \end{proof}

  \begin{remark}
    It can also be decided in polynomial time whether a distribution
  $x \in \Chip(G)$ is under an acyclic orientation. 
  A greedy algorithm solves the problem.
  \end{remark}


  From the previous proposition, using the duality between 
  chip-distributi\-ons and divisors, we get a special case of the 
  Riemann-Roch theorem for graphs.
  
  Let us denote by $K$ the canonical divisor on a graph $G$, that is,
  $K(v)=d(v)-2$ for each vertex $v\in V(G)$.  
  
  \begin{thm}[Riemann-Roch for graphs, \cite{BN-Riem-Roch}]
   Let $G$ be a graph, and let $f$ be a divisor on $G$. Then
   $$\rank(f)-\rank(K-f)=\deg(f)-|E(G)|+|V(G)|.$$
  \end{thm}

  Now, from Proposition \ref{prop::rank_of_dual} and Proposition 
  \ref{prop::acikl_alatt} we have for $f=K^+ - x$ that 
  $$\rank(f) = \dist(x) - 1 = |E(G)| - |x| - 1 = 
  \deg(f)- |E(G)| + |V(G)| -1,$$
  if $x$ is under an acyclic orientation.
  
  We claim that in this case, $\rank(K-f)=-1$.
  Indeed, $K-f=K-(K^+ - x)= x - \mathbf{1}$, so the dual of $K-f$ is
  $K^+ - x + \mathbf{1}$. The distribution $x$ is under an acyclic orientation, let $D$ be an orientation witnessing this, i.e., 
  $x(v) \le d^-_D(v)$ for each vertex $v \in V(G)$. 
  Then $(K^+ - x + \mathbf{1})(v) = d(v) - x(v) \ge d^+_D(v)$ 
  for each vertex $v \in V(G)$, hence we can use Lemma \ref{lem::chip_dist_over_acyclic} for $K^+ - x + \mathbf{1}$ and 
  the directed graph obtained from $D$ by reversing every edge. 
  It follows that $K^+ - x + \mathbf{1}$ is non-terminating, 
  hence for its dual, $\rank(K-f)=-1$ by Proposition \ref{prop::rank_of_dual}. 

  Therefore, we have $\rank(f) - \rank(K-f)= \deg(f)- |E(G)| + |V(G)|$, 
  showing the Riemann-Roch theorem in this special case.

\section*{Acknowledgement}
  Research was supported by the MTA-ELTE Egerv\'ary Research Group and by the
Hungarian Scientific Research Fund - OTKA, K109240 (Lilla T\'oth\-m\'e\-r\'esz), 104178 and 105645 (Viktor Kiss).

  We would like to thank B\'alint Hujter for introducing us to this topic, Erika B\'erczi-Kov\'acs and Krist\'of B\'erczi for suggesting us to try to reduce the feedback arc set problem to the computation of the rank. We would also like to thank Tam\'as Kir\'aly, Zolt\'an Kir\'aly, M\'arton Elekes and the anonymous referees for their useful comments about the manuscript.

\bibliographystyle{abbrv}
\bibliography{gon_rank}

\end{document}